\documentclass{IEEEtran}

\usepackage{amsmath,amssymb, amsthm,cite,graphicx}

\long\def\symbolfootnote[#1]#2{\begingroup
\def\thefootnote{\fnsymbol{footnote}}\footnote[#1]{#2}\endgroup}

\newtheorem{theorem}{Theorem}

\newcommand{\In}{\mathsf{In}}
\newcommand{\Mid}{\mathsf{Mid}}
\newcommand{\Out}{\mathsf{Out}}

\title{Cooperative Regenerating Codes for \\ Distributed Storage Systems}
\author{Kenneth W. Shum, \IEEEmembership{Member, IEEE}\thanks{K. Shum is with the Institute of Network Coding, the Chinese University of Hong Kong, Shatin, Hong Kong. Email: wkshum at inc.cuhk.edu.hk.}}

\begin{document}

\maketitle

\begin{abstract}
\symbolfootnote[0]{This work was partially supported by a grant from the University Grants Committee of the Hong Kong Special Administrative Region, China (Project No. AoE/E-02/08).} When there are multiple node failures in a distributed storage system, regenerating the failed storage nodes individually in a one-by-one manner is suboptimal as far as repair-bandwidth minimization is concerned. If data exchange among the newcomers is enabled, we can get a better tradeoff between repair bandwidth and the storage per node. An explicit and optimal construction of cooperative regenerating code is illustrated.
\end{abstract}


\section{Introduction}
Distributed storage system provides a scalable solution to the ever-increasing demand of reliable storage. The storage nodes are distributed in different geographical locations, and in case some disastrous event happened to one of them, the source data would remain intact.  There are two common strategies for preventing data loss against storage node failures. The first one, employed by the current Google file system~\cite{GFS}, is {\em data replication}. Although replication-based scheme is easy to manage, it has the drawback of low storage efficiency. The second one is based on {\em erasure coding}, and is used in Oceanstore~\cite{Oceanstore} and Total Recall~\cite{Totalrecall} for instance. With erasure coding, The storage network can be regarded as an erasure code which can correct any $n-k$ erasures; a file is encoded into $n$ pieces of data, and from any $k$ of them the original file can be reconstructed.

When a storage node fails, an obvious way to repair it is to rebuild the whole file from some other $k$ nodes, and then re-encode the data. The disadvantage of this method is that,  when the file size is very large, excessive traffic is generated in the network. The bandwidth required in the repairing process seems to be wasted, because only a fraction of the downloaded data is stored in the new node after regeneration. By viewing the repair problem as a single-source multi-cast problem in network coding theory, Dimakis {\em et al.} discovered a tradeoff between the amount of storage in each node and the bandwidth required in the repair process~\cite{DGWR07}.  Erasure codes for distributed storage system, aiming at minimizing the repair-bandwidth, is called {\em regenerating code}. The construction of regenerating code is under active research. We refer the readers to~\cite{DGWR2010} and the references therein for the application of network coding in distributed storage systems.

Most of the results in the literature on regenerating codes are for repairing a single storage node. However, there are several scenarios where {\em multiple} failures must be considered.
Firstly, in a system with high churn rate, the nodes may join and leave the system very frequently. When two or more nodes join the distributed storage system at the same time, the new nodes can exploit the opportunity of exchange data among themselves in the repair process. Secondly, node repair may be done in batch. In systems like Total Recall, a recovery is triggered when the fraction of available nodes is below a certain threshold, and the failed nodes are then repaired as a group. The new nodes which are going to be regenerated are called {\em newcomers}. There are two ways in regenerating a group of newcomer: we may either repair them one by one, or repair them jointly with cooperation among the newcomers. It is shown in~\cite{HXWZL10, WXHO10} that further reduction of repair-bandwidth is possible with cooperative repair. Let the number of newcomers be~$r$. In \cite{HXWZL10} each newcomer is required to connect to all $n-r$ surviving storage nodes during the repair process, and in~\cite{WXHO10}, this requirement is relaxed such that different newcomers may have different number of connections. However, in both~\cite{HXWZL10} and \cite{WXHO10}, only the storage systems which minimize storage per node are considered.

In this paper, an example of cooperatively regenerating multiple newcomers is described in Section~\ref{sec:coop}. In Section~\ref{sec:mincut}, we define the information flow graph for cooperative repair, and derive a lower bound on repair-bandwidth. This lower bound is applicable to {\em functional repair}, where the content of a newcomer may not be the same as the failed node to be replaced, but the property that any $k$ nodes can reconstruct the original file is retained. The lower bound is function of the storage per node, and hence is an extension of the results in~\cite{HXWZL10}. A more practical and easier-to-manage mode of operation is called {\em exact repair}, in which the regenerated node contains exactly the same encoded data as in the failed node. In Section~\ref{sec:explicit},  we give a family of explicit code constructions which meet the lower bound, and hence show that the construction is optimal.

\section{An Example of Cooperative Repair} \label{sec:coop}

Consider the following example taken from~\cite{WD09}. Four data packets $A_1$, $A_2$, $B_1$ and $B_2$, are distributed to four storage nodes.  Each of them stores two packets. The first one stores $A_1$ and $A_2$, the second stores $B_1$ and $B_2$. The third and fourth nodes are parity nodes. The third node contains two packets $A_1+B_1$ and $2A_2+B_2$, and the last node contains $2A_1+B_1$ and $A_2+B_2$. Here, a packet is interpreted as an element in a finite field, and addition and multiplication are finite field operations. We can take $GF(5)$ as the underlying finite field in this example. Any data collector, after  downloading the packets from any two storage nodes, can reconstruct the four original packets by solving a system of linear equations. For example, if we download from the third and fourth nodes, we can recover $A_1$ and $B_1$ from packets $A_1+B_1$ and $2A_1+B_1$, and recover $A_2$ and $B_2$ from packets $2A_2+B_2$ and $A_2+B_2$.

Suppose that the first node fails.  To repair the first node, we can download four packets from any other two nodes, from which we can recover the two required packets $A_1$ and~$A_2$. For example, if we download the packets from the second and third nodes, we have $B_1$, $B_2$, $A_1+B_1$ and $2 A_2+B_2$. We can then recover $A_1$ by subtracting $B_1$ from $A_1+B_1$, and $A_2$ by computing $((2A_2+B_2) - B_2 )/2$. It is illustrated in~\cite{WD09} that we can reduce the repair-bandwidth from four packets to three packets, by making three connections to the three remaining nodes, and downloading one packet from each of them (Fig.~\ref{fig:IA}). Each of the three remaining nodes simply adds the two packets and sends the sum to the newcomer, who can then subtract off $B_1+B_2$ and obtain $A_1+2A_2$ and $2A_1+A_2$, from which $A_1$ and $A_2$ can be solved.  

\begin{figure}
\centering
\includegraphics[width=2.5in]{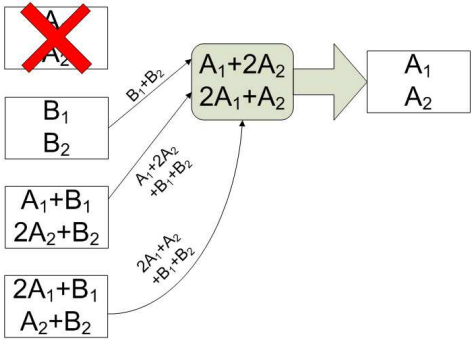}
\caption{Repairing a single node failure with minimum repair bandwidth} \label{fig:IA}
\end{figure}

When two storage nodes fail simultaneously, the computational trick mentioned in the previous paragraph no longer works. Suppose that the second and the fourth storage nodes fail at the same time. To repair both of them separately, each of the newcomers can download four packets from the remaining storage nodes, reconstruct packets $A_1$, $A_2$, $B_1$ and $B_2$, and re-encode the desired packets (Fig.~\ref{fig:SR}). This is the best we can do with separate repair. Using the result in~\cite{DGWR07}, it can be shown that any one-by-one repair process with repair-bandwidth strictly less than four packets per newcomer is infeasible.

If the two newcomers can exchange data during the regeneration process, the total repair-bandwidth can indeed be reduced from eight packets to six packets (Fig.~\ref{fig:CR}). The two newcomers first make an agreement that one of them downloads the packets with subscript 1, and the other one downloads the packets with subscript~2. (They can compare, for instance, their serial numbers in order to determine who downloads the packets with smaller subscript.) The first newcomer gets $A_1$ and $A_1+B_1$ from node 1 and 3 respectively, while the second newcomer gets $A_2$ and $2A_2+B_2$ from node 1 and 3 respectively.  The first newcomer then computes $B_1$ and $2A_1+B_1$ by taking the difference and the sum of the two inputs. The packet $B_1$ is stored in the first newcomer and $2A_1+B_1$ is sent to the second newcomer. Similarly, the second newcomer computes $B_2$ and $A_2+B_2$, stores $A_2+B_2$ in memory and sends $B_2$ to the first newcomer. Only six packet transmissions are required in this joint regeneration process.

\begin{figure}
\centering
\includegraphics[width=2.5in]{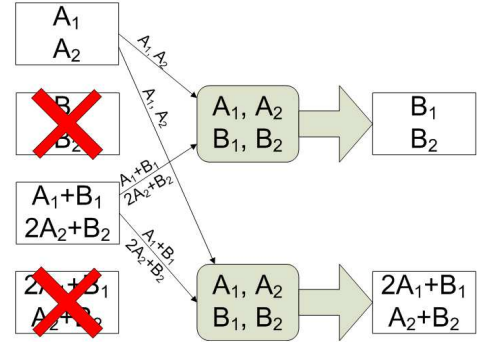}
\caption{Individual regeneration of multiple failures} \label{fig:SR}
\end{figure}

\begin{figure}
\centering
\includegraphics[width=2.5in]{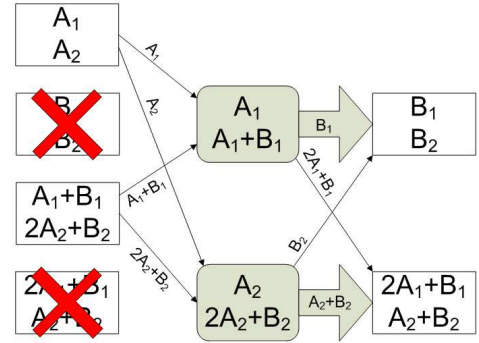}
\caption{Cooperative regeneration of multiple failures} \label{fig:CR}
\end{figure}

\section{Information Flow Graph and Min-Cut Bound} \label{sec:mincut}

We formally define the cooperative repair problem as follows. There are two kinds of entities in a distributed storage system, {\em storage nodes} and {\em data collectors}, and two kinds of operations, {\em file reconstruction} and {\em node repair}. A file of size $B$ units is encoded and distributed among the $n$ storage nodes, each of them stores $\alpha$ units of data. The file can be reconstructed by a data collector connecting to any $k$ storage nodes. Upon the failure of $r$ nodes, a two-phase repair process is triggered.
In the first phase, each of the $r$ newcomers connects to $d$ remaining storage nodes, and download $\beta_1$ units of data from each of them. After processing the data they have downloaded, the $r$ newcomers exchange some data among themselves, by sending $\beta_2$ units of data to each of the other $r-1$ newcomers. Each newcomer downloads $d\beta_1$ units of data in the first phase and  $(r-1)\beta_2$ units of data in the second phase. The repair-bandwidth per node is thus
$\gamma = d\beta_1 + (r-1) \beta_2$.

In the remaining of this paper, we will assume that $d\geq k$.

\begin{figure}
\centering
\includegraphics[width=2.5in]{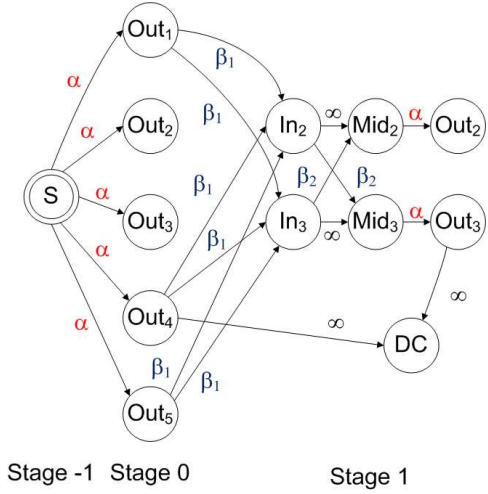}
\caption{Information flow graph} \label{fig:flow1}
\end{figure}

We construct an {\em information flow graph} as follows. There are three types of vertices in the information flow graph: one for the source data, one for the storage nodes and one for data collectors. The vertices are divided into stages. We proceed from one stage to the next stage after a  repair process is completed. (Fig.~\ref{fig:flow1}).

There is one single vertex, called the {\em source vertex}, in stage $-1$, representing the original data file. The $n$ storage nodes are represented by $n$ vertices in stage~0, called $\Out_i$, for $i=1,2,\ldots, n$. The source vertex is connected to each vertex in stage 0 by a directed edge with capacity~$\alpha$. For $s = 1,2,3,\ldots$, let $\mathcal{R}_{s}$ be the set of $r$ storage nodes which fail in stage $s-1$, and are regenerated in stage~$s$. The set $\mathcal{R}_{s}$ is a subset of $\{1,2,\ldots, n\}$ with cardinality~$r$. For each storage node $p$ in $\mathcal{R}_{s}$, we construct three vertices in stage $s$: $\In_p$, $\Mid_p$ and $\Out_p$. Vertex $\In_p$ has $d$ incoming edges with capacity $\beta_1$, emanated from $d$ ``out'' nodes in previous stages. We join vertex $\In_p$ and  $\Mid_p$ with a directed edge of infinite capacity. For $p, q \in \mathcal{R}_{s}$, $p\neq q$, there is a directed edge from $\In_p$ to $\Mid_q$  with capacity $\beta_2$.  Newcomer $p$ stores $\alpha$ units of data, and this is represented by a directed edge from $\Mid_p$ to $\Out_p$ with capacity~$\alpha$.

For each  data collector, we add a vertex, called $\mathsf{DC}$,  in the information flow graph. It is connected to $k$ ``out'' nodes with distinct indices, but not necessarily from the same stage, by  $k$ infinite-capacity edges.

We call an information flow graph constructed in this way $G(n,k,d,r; \alpha, \beta_1, \beta_2)$, or simply $G$ if the parameters are clear from the context. The number of stages is potentially  unlimited.

A {\em cut} in an information flow graph is a partition of the set of vertices, $(\mathcal{U}, \bar{\mathcal{U}})$,  such that the source vertex is in $\mathcal{U}$ and a designated data collector is in $\bar{\mathcal{U}}$. We associate with each cut a value, called the {\em capacity},  defined as the sum of the capacities of the directed edges from vertices in $\mathcal{U}$ to vertices in $\bar{\mathcal{U}}$. An example is shown in Fig.~\ref{fig:cut}.
The max-flow-min-cut bound in network coding for single-source multi-cast network states that if the minimum cut capacities between data collectors and the source is at no larger than $C$, then the amount of data we can send to each data collector is no more than $C$~\cite{ACLY00}.

\begin{theorem} Suppose that $d\geq k$. The minimum cut of an information flow graph $G$ is less than or equal to
\begin{equation}
\sum_{i=1}^k \ell_i \min\Big\{\alpha, \Big(d - \sum_{j=1}^{i-1}\ell_j\Big) \beta_1 + (r-\ell_i) \beta_2 \Big\}
\label{flow0}
\end{equation}
where  $(\ell_1, \ell_2, \ldots , \ell_k)$ is any
$k$-tuple of integers satisfying $\ell_1 + \ell_2+\ldots + \ell_k = k$ and $0\leq \ell_i \leq r$ for all~$i$.
\label{thm:mincut}
\end{theorem}

\begin{proof}
By relabeling the nodes if necessary, suppose that a data collector $\mathsf{DC}$ connects to storage node 1 to node~$k$. Let $s_1 < s_2 < \ldots < s_m$ be the stages in which nodes 1 to $k$ are most recently repaired, where $m$ is an integer. We note that
$\{1,2,\ldots, k\}$ is contained in the union of $\mathcal{R}_{s_1}$, $\mathcal{R}_{s_2}, \ldots,  \mathcal{R}_{s_m}$.
For $i=1,2,\ldots, m$, let
$$\mathcal{S}_{i} := \big(\{1,2,\ldots,k\}\cap \mathcal{R}_{s_i} \big)  \setminus (\mathcal{R}_{s_{i+1}} \cup \cdots \cup \mathcal{R}_{s_m}).$$
The physical meaning of $\mathcal{S}_i$ is that the storage nodes with indices in $\mathcal{S}_i$ are repaired in stage $s_i$ and remain intact until the data collector $\mathsf{DC}$ shows up. The index sets $\mathcal{S}_i$'s are disjoint and their union is equal to $\{1,2,\ldots, k\}$.
We let $\ell_i$ to be the cardinality of $\mathcal{S}_i$. Obviously we have $\ell_1+\ell_2+\ldots+\ell_m= k$, $\ell_i\leq r$ for all~$i$, and $m\leq k$.

\begin{figure}
\centering
\includegraphics[width=3.2in]{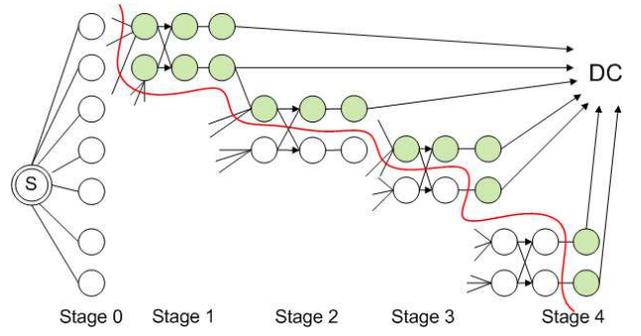}
\caption{A sample cut in the information flow graph.} \label{fig:cut}
\end{figure}

\begin{figure}
\centering
\includegraphics[width=2.2in]{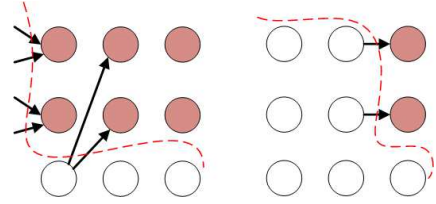}
\caption{Two different kinds of cuts within a stage.} \label{fig:cut2}
\end{figure}

For $i=1,2,\ldots, m$,  the $\ell_i$ ``out'' nodes in stage $s_i$ which are connected directly to $\mathsf{DC}$ must be in $\bar{\mathcal{U}}$, otherwise, there would be an infinite-capacity edge from $\mathcal{U}$ to $\bar{\mathcal{U}}$. In stage $s_i$, we consider two different ways to construct a cut. We either put all ``in'' and ``mid'' nodes associated to the storage nodes in $\mathcal{S}_i$ in $\bar{\mathcal{U}}$, or put all of them in $\mathcal{U}$. In Fig.~\ref{fig:cut2}, we graphically illustrate the two different cuttings. The shaded vertices are in $\bar{\mathcal{U}}$ and the edges from $\mathcal{U}$ to $\bar{\mathcal{U}}$ are shown.

Each ``in'' node in the first cut may connect to as small as $d - \sum_{j=1}^{i-1} \ell_j$ ``out'' nodes in $\mathcal{U}$ in previous stages. The sum of edge capacities from $\mathcal{U}$ to $\bar{\mathcal{U}}$ can be as small as
$\ell_i (d - \sum_{j=1}^{i-1} \ell_j) \beta_1 + (r-\ell_i)\ell_i \beta_2$.
In the second kind of cut, the sum of edge capacities from $\mathcal{U}$ to $\bar{\mathcal{U}}$ is $\ell_i \alpha$. After taking the minimum of these two cut values, we get
\begin{equation}
 \ell_i\min\Big\{ \alpha,  \Big(d-\sum_{j=1}^{i-1} \ell_j\Big)\beta_1 + (r-\ell_i) \beta_2 \Big\}. \label{flow2}
\end{equation}
We obtain the expression in~\eqref{thm:mincut} by summing~\eqref{flow2} over $i=1,2,\ldots, m$.
\end{proof}

A cut described in the proof of Theorem~\ref{thm:mincut} is called a cut of type $(\ell_1,\ell_2,\ldots, \ell_k)$.

We illustrate Theorem~\ref{thm:mincut} by  the example in Section~\ref{sec:coop}. The parameters are $n=4$, $d=k=r=2$, $B=4$, and $\alpha = B/k = 2$. The are two pairs of integers $(\ell_1, \ell_2)$, namely $(2,0)$ and $(1,1)$,  which satisfy the condition in Theorem~\ref{thm:mincut}.  The capacity of minimum cut, by Theorem~\ref{thm:mincut},
is no more than $2\min\{\alpha,2\beta_1\}$ and
$\min\{\alpha, 2\beta_1 + \beta_2\}+\min\{\alpha, \beta_1 + \beta_2\} $. The first cut imposes the upper bound $B \leq 2 \min\{\alpha, 2\beta_1)$ on the file size~$B$, which implies that $\beta_1 \geq 1$. The second cut imposes another constraint on $B$,
\[
4 \leq \min\{2, 2\beta_1 + \beta_2\}+\min\{2, \beta_1 + \beta_2\},
\]
from which we can deduce that $\beta_1+\beta_2 \geq 2$. After summing $\beta_1 \geq 1$ and $\beta_1+\beta_2 \geq 2$, we obtain $\gamma = 2\beta_1 + \beta_2 \geq 3$. The minimum possible repair-bandwidth $\gamma=3$
matched by the regenerating code presented in Section~\ref{sec:coop}. The regenerating code in Section~\ref{sec:coop} is therefore optimal.

\begin{figure}
\centering
\includegraphics[width=3in]{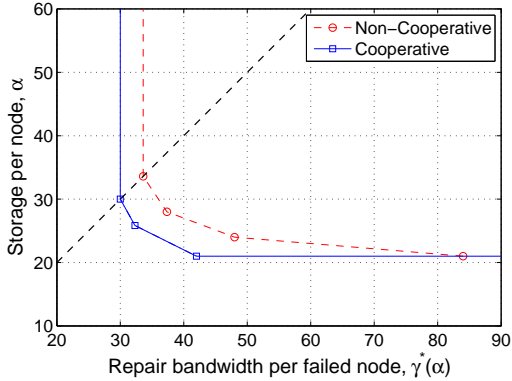}
\caption{Lower bound on repair-bandwidth ($B=84$, $d=4$, $k=4$, $r=3$)} \label{fig:tradeoff1}
\end{figure}

We can formulate the repair-bandwidth minimization problem as follows. Given the storage per node, $\alpha$, we want to minimize the objective function $\gamma=d\beta_1 +(r-1) \beta_2$ over all non-negative $\beta_1$ and $\beta_2$ subject to  the constraints that the file size $B$ is no more than the values in~\eqref{thm:mincut}, for all legitimate $(\ell_1,\ell_2,\ldots,\ell_k)$. It can be shown that the minimization problem can be reduced to a linear program, and hence can be effectively solved.
We let the resulting optimal value be denoted by $\gamma^*( \alpha)$. This is a lower bound on repair-bandwidth for a given value of $\alpha$.

In Fig.~\ref{fig:tradeoff1}, we illustrate the lower bound $\gamma^*(\alpha)$ for $B=84$, $d=4$, $k=4$ and $r=3$. For comparison, we plot the storage-repair-bandwidth tradeoff for non-cooperative one-by-one repair in Fig.~\ref{fig:tradeoff1}. From~\cite[Theorem 1]{DGWR2010}, the smallest
repair-bandwidth of a non-cooperative minimum-storage regenerating code is given by the formula $Bd/(k(d-k+1))$, which is equal to 84 in this example.  It can be shown that $\gamma^*(B/k) = B(d+r-1)/(k(d+r-k))$. In the next section, we give a construction of cooperative regenerating code which meets the lower bound $\gamma^*(B/k)$ when $d=k$.

\section{An Explicit Construction for Exact Repair}
\label{sec:explicit}
Exact repair has the advantage that the encoding vectors of the newcomers remain the same. This helps in reducing maintenance overhead. For non-cooperative and one-by-one repair, there are several exact constructions of regenerating code available in the literature, for example the constructions in~\cite{SR_isit2010} and~\cite{SRK10a}. In this section,  we construct a family of regenerating codes for cooperative repair with parameters $d=k \leq n-r$, which contains the example given in Section~\ref{sec:coop} as special case.

The recipe of this construction needs an maximal-distance separable (MDS) code of length~$n$ and dimension~$k$. Given $n$, let $q$ be the smallest prime power larger than or equal to~$n$. We use the Reed-Solomon (RS) code over $GF(q)$ generated by the following generator matrix
\[ \mathbf{G} := {\small
\begin{bmatrix}
1   & 1   & 1   & 1&   \ldots & 1 &1 \\
a_1 & a_2 & a_3 & a_4& \ldots & a_{n-1}& a_n \\
\vdots & \vdots & \vdots&\vdots & \ddots & \vdots& \vdots \\
a_1^{k-1}&a_2^{k-1}&a_3^{k-1} & a_4^{k-1}&\ldots &a_{n-1}^{k-1}& a_n^{k-1}
\end{bmatrix} }
\]
where $a_1$, $a_2,\ldots, a_n$ are $n$ distinct elements in $GF(q)$. Let $\mathbf{g}_i$ be the $i$th column of~$\mathbf{G}$.  Given $k$ message symbols in $GF(q)$, we put them in a row vector $\mathbf{m}^T = [m_1\ m_2\ \ldots\ m_k]$. (The superscript ``$^T$'' is the transpose operator.) We encode $\mathbf{m}^T$ into the codeword $\mathbf{m}^T \mathbf{G}$.  The MDS property of RS code follows from the fact that every $k\times k$ submatrix of $\mathbf{G}$ is a non-singular Vandermonde matrix.

We apply the technique called ``striping'' from coding for disk arrays. The whole file of size $B$ is divided into many stripes, or chunks, and each chunk of data is encoded and treated in the same way. In the following, we will only describe the operations on each stripe of data.

We divide a stripe of data into $kr$ packets, each of them is considered as an element in $GF(q)$. The $kr$ packets are laid out in an $r \times k$ matrix $\mathbf{M}$, called the {\em message matrix}.
To set up the distributed storage system, we first encode the message matrix $\mathbf{M}$ into
$\mathbf{MG}$, which is an $r\times n$ matrix. For $j=1,2,\ldots, n$, node $j$ stores the $r$ packets in the $j$th column of $\mathbf{MG}$. Let the $r$ rows of $\mathbf{M}$ be denoted by $\mathbf{m}_1^T$, $\mathbf{m}_2^T, \ldots, \mathbf{m}_r^T$.
The packets stored in node $j$ are $\mathbf{m}_i^T \mathbf{g}_j$, for $i=1,2,\ldots, r$.

A data collector downloads from $k$ storage nodes, say nodes $c_1$, $c_2, \ldots, c_k \in\{1,2,\ldots, n\}$. The $kr$ received packets are arranged in an $r\times k$ matrix. The $(i,j)$-entry of this matrix is $\mathbf{m}_i^T \mathbf{g}_{c_j}$. This matrix can be factorized  as
$\mathbf{M} \cdot [
\mathbf{g}_{c_1}\  \mathbf{g}_{c_2}\  \cdots\  \mathbf{g}_{c_k}]$. We can reconstruct the original file  by inverting the Vandermonde matrix $[\mathbf{g}_{c_1}\ \mathbf{g}_{c_2}\ \cdots \ \mathbf{g}_{c_k}]$.

Suppose that nodes $f_1$, $f_2,\ldots, f_r$ fail.  The $r$ newcomers first coordinate among themselves, and agree upon an order of the newcomers, say by their serial numbers. For the ease of notation, suppose that newcomer $f_j$ is the $j$th newcomer, for $j=1,2,\ldots, r$. The $j$th newcomer $f_j$ connects to any other $k$ remaining storage nodes, say $\pi_j(1), \pi_j(2), \ldots, \pi_j(k)$, and downloads the packets encoded from $\mathbf{m}_j^T$, namely, $\mathbf{m}_j^T \mathbf{g}_{\pi_j(1)}$, $\mathbf{m}_j^T \mathbf{g}_{\pi_j(2)}, \ldots, \mathbf{m}_j^T \mathbf{g}_{\pi_j(k)}$. (Recall that we assume $k=d$ in this construction.) Since $[\mathbf{g}_{\pi_j(1)} \ \mathbf{g}_{\pi_j(2)} \ldots \mathbf{g}_{\pi_j(k)}]$ is non-singular, newcomer $f_j$ can recover the message vector $\mathbf{m}_j^T$ after the first phase. In the second phase, newcomer $f_j$ computes  $\mathbf{m}_j^T \mathbf{g}_{f_i}$ for $i = 1,2, \ldots, r$, and sends the packet $\mathbf{m}_j^T \mathbf{g}_{f_i}$ to newcomer~$f_i$, $i\neq j$. A total of $r-1$ packets are sent from each newcomer in the second phase. After the exchange of packets, newcomer $f_j$ then has  the $r$ required packets $\mathbf{m}_i^T \mathbf{g}_{f_j}$, for $i=1,2,\ldots, r$. The repair-bandwidth per each newcomer is $k+r-1 = d+r-1$.

In this construction, we can pick the smallest prime power $q$ larger than or equal to $n$ as the size of the finite field. If the number of storage nodes $n$ increases, the finite field size increases linearly with~$n$.


\begin{theorem}
The cooperative regenerating code described above is optimal, in the sense that if $B=kr$, $k=d$, and each node stores $\alpha = r$ packets, the minimal repair-bandwidth  per each failed node is equal to $k+r-1$. \label{thm:MBSR}
\end{theorem}

\begin{proof} We use the notation as in Theorem~\ref{thm:mincut}. 
The capacity of a cut of type $(\ell_1, \ell_2, \ldots, \ell_k)$, as shown in~\eqref{flow0}, is an upper bound on $kr$.
If any summand $(d-\sum_{j=1}^{i-1} \ell_j)\beta_1+(r-\ell_i)\beta_2$ in~\eqref{flow0} is strictly less than $\alpha=B/k=r$ for any $i$, then the value in~\eqref{flow0} is strictly less than $\sum_{i=1}^k \ell_i r = kr$.
This would violate the fact that $kr$ is upper bounded by~\eqref{flow0}. Hence we have \begin{equation}
(k-\sum_{j=1}^{i-1} \ell_j)\beta_1+(r-\ell_i)\beta_2 \geq  B/k = r \label{eq:cut2}
\end{equation}
for any cut associated with $(\ell_1, \ell_2, \ldots, \ell_k)$ and any $i$.

{\em Case 1: $r \leq k = d$}. From a cut of type $(\ell_1,\ell_2,\ldots, \ell_k) = (1,1,\ldots, 1)$, we have
\begin{equation}
 \beta_1 + (r-1) \beta_2 \geq r \label{eq:cut3}
\end{equation}
from \eqref{eq:cut2}. From another cut of type 
$(\ell_1,\ell_2, \ldots, \ell_k) = (\underbrace{1,1,\ldots,1}_{k-r},r, 0,\ldots)$, from \eqref{eq:cut2} again, we obtain the condition
$$(k-(k-r))\beta_1 + (r-r) \beta_2  = r \beta_1 \geq r$$
which implies that $\beta_1 \geq 1$. We then add $(k-1)\beta_1 \geq k-1$ to \eqref{eq:cut3}, and get $\gamma = k\beta_1 + (r-1)\beta_2 \geq k+r-1$.

{\em Case 2: $r > k = d$}. Consider the two cuts associated with $(\ell_1,\ell_2,\ldots, \ell_k)$ equal to $(k,0,\ldots, 0)$ and $(k-1,1,0,\ldots, 0)$. We obtain the following two inequalities from~\eqref{eq:cut2},
\begin{align}
k \beta_1 + (r-k) \beta_2 \geq r \label{eq:a}\\
\beta_1 + (r-1) \beta_2 \geq r. \label{eq:b}
\end{align}
We multiply both sides of~\eqref{eq:a} by $(r-1)$, and multiply both sides of~\eqref{eq:b} by $k$.  After adding the two resulting inequalities, we  get $\gamma = k\beta_1 + (r-1)\beta_2 \geq k+r-1$.

The repair-bandwidth per failed node is therefore cannot be less than $k+r-1$. The repair-bandwidth of the code constructed in this section matches this lower bound, and is hence optimal.
\end{proof}

The regenerating code constructed in this section has the advantage that a storage node participating in a regeneration process is required to read and exactly the same amount of data to be sent out, without any arithmetical operations. This is called  the {\em uncoded repair}  property~\cite{ElRouayheb10}.  


We compare below the repair-bandwidth of three different modes of repair, all with parameters $n=7$, $B=84$, $k=4$ and $\alpha=B/4 = 21$. Suppose that three nodes fail simultaneously.

(i) Individual repair without newcomer cooperation. Each newcomer connects to the four remaining storage nodes. As calculated in the previous section, the repair-bandwidth per newcomer is $84$.

(ii) One-by-one repair utilizing the newly regenerated node as a helper. The average repair-bandwidth per newcomer is
\[
 \frac{1}{3} \Big( \frac{84(4)}{4(4-4+1)} + \frac{84(4)}{4(5-4+1)} +\frac{84(4)}{4(6-4+1)} \Big)=51.333.
\]
The first term in the parenthesis is the repair-bandwidth of the first newcomer, which downloads from the four surviving nodes, the second term is the repair-bandwidth of the second newcomer, who connects to the four surviving nodes and the newly regenerated newcomer, and so on.

(iii) Full cooperation among the three newcomers. The repair-bandwidth per newcomer can be reduced to~42 using the regenerating code given in this section. We thus see that newcomer cooperation is able to reduce the repair-bandwidth of a distributed storage system significantly.



\begin{thebibliography}{10}
\providecommand{\url}[1]{#1}
\csname url@samestyle\endcsname
\providecommand{\newblock}{\relax}
\providecommand{\bibinfo}[2]{#2}
\providecommand{\BIBentrySTDinterwordspacing}{\spaceskip=0pt\relax}
\providecommand{\BIBentryALTinterwordstretchfactor}{4}
\providecommand{\BIBentryALTinterwordspacing}{\spaceskip=\fontdimen2\font plus
\BIBentryALTinterwordstretchfactor\fontdimen3\font minus
  \fontdimen4\font\relax}
\providecommand{\BIBforeignlanguage}[2]{{%
\expandafter\ifx\csname l@#1\endcsname\relax
\typeout{** WARNING: IEEEtran.bst: No hyphenation pattern has been}%
\typeout{** loaded for the language `#1'. Using the pattern for}%
\typeout{** the default language instead.}%
\else
\language=\csname l@#1\endcsname
\fi
#2}}
\providecommand{\BIBdecl}{\relax}
\BIBdecl

\bibitem{GFS}
S.~Ghemawat, H.~Gobioff, and S.-T. Leung, ``The {G}oogle file system,'' in
  \emph{Proc. of the 19th {ACM} {SIGOPS} Symp. on Operating Systems Principles
  (SOSP'03)}, Oct. 2003.

\bibitem{Oceanstore}
{J. Kubiatowicz et al.}, ``Ocean{S}tore: an architecture for global-scale
  persistent storage,'' in \emph{Proc. 9th Int. Conf. on Architectural Support
  for programming Languages and Operating Systems (ASPLOS)}, Cambridge, MA,
  Nov. 2000, pp. 190--201.

\bibitem{Totalrecall}
R.~Bhagwan, K.~Tati, Y.~Cheng, S.~Savage, and G.~Voelker, ``Total recall:
  system support for automated availability management,'' in \emph{Proc. of the
  1st Conf. on Networked Systems Design and Implementation}, San Francisco,
  Mar. 2004.

\bibitem{DGWR07}
A.~G. Dimakis, P.~B. Godfrey, Y.~Wu, M.~J. Wainwright, and K.~Ramchandran,
  ``Network coding for distributed storage system,'' in \emph{Proc. {IEEE} Int.
  Conf. on Computer Commun. (INFOCOM '07)}, Anchorage, Alaska, May 2007.

\bibitem{DGWR2010}
------, ``Network coding for distributed storage systems,'' \emph{{IEEE} Trans.
  Inf. Theory}, vol.~56, no.~9, pp. 4539--4551, Sep. 2010.

\bibitem{HXWZL10}
Y.~Hu, Y.~Xu, X.~Wang, C.~Zhan, and P.~Li, ``Cooperative recovery of
  distributed storage systems from multiple losses with network coding,''
  \emph{{IEEE} J. on Selected Areas in Commun.}, vol.~28, no.~2, pp. 268--275,
  Feb. 2010.

\bibitem{WXHO10}
X.~Wang, Y.~Xu, Y.~Hu, and K.~Ou, ``{MFR}: Multi-loss flexible recovery in
  distributed storage systems,'' in \emph{Proc. IEEE Int. Conf. on Comm.
  (ICC)}, Capetown, South Africa, May 2010.

\bibitem{WD09}
Y.~Wu and A.~G. Dimakis, ``Reducing repair traffic for erasure coding-based
  storage via interference alignment,'' in \emph{Proc. {IEEE} Int. Symp. Inf.
  Theory}, Seoul, Jul. 2009.

\bibitem{ACLY00}
R.~Ahlswede, N.~Cai, S.-Y.~R. Li, and R.~W. Yeung, ``Network information
  flow,'' \emph{{IEEE} Trans. Inf. Theory}, vol.~46, pp. 1204--1216, 2000.

\bibitem{SR_isit2010}
C.~Suh and K.~Ramchandran, ``Exact-repair {MDS} codes for distributed storage
  using interference alignment,'' in \emph{Proc. {IEEE} Int. Symp. Inf.
  Theory}, Austin, Jun. 2010, pp. 161--166.

\bibitem{SRK10a}
N.~B. Shah, K.~V. Rashmi, and P.~V. Kumar, ``A flexible class of regenerating
  codes for distributed storage,'' in \emph{Proc. {IEEE} Int. Symp. Inf.
  Theory}, Austin, Jun. 2010, pp. 1943--1947.

\bibitem{ElRouayheb10}
S.~{El Rouayheb} and K.~Ramchandran, ``Fractional repetition codes for repair
  in distributed storage systems,'' in \emph{Allerton conference on commun.
  control and computing}, Monticello, Sep. 2010.

\end{thebibliography}


\end{document}